\documentclass[12pt,a4paper]{article}
\usepackage[T1]{fontenc}
\usepackage{lmodern}
\usepackage{amsmath,amssymb,amsthm,mathtools,bm}
\usepackage{microtype}
\usepackage{geometry}
\usepackage{hyperref}
\usepackage{cite}
\usepackage{booktabs}
\usepackage{xcolor}
\usepackage{tikz}
\hypersetup{colorlinks=true,linkcolor=blue!55!black,citecolor=blue!55!black,urlcolor=blue!55!black}

\newcommand{\R}{\mathbb R}
\newcommand{\Z}{\mathbb Z}
\newcommand{\dd}{\mathrm d}
\newcommand{\e}{\mathrm e}

\newcommand{\AdS}{\mathrm{AdS}}
\newcommand{\cM}{\mathcal M}
\newcommand{\cQ}{\mathcal Q}
\newcommand{\cO}{\mathcal O}
\newcommand{\cH}{\mathcal H}
\newcommand{\bulk}{\mathrm{bulk}}

\newcommand{\bdy}{\mathrm{bdy}}
\newcommand{\ren}{\mathrm{ren}}

\newtheorem{proposition}{Proposition}

\title{\Large \textbf{Type IIB Axion--Dilaton Wormholes \\ and \\ the BPS Limit Hessian}}
\author{Soo-Jong Rey\\[0.2cm]
{\sl Kwangwoon University} \\
{\sl Seoul, Korea}}
\date{}
\begin{document}
\maketitle

\begin{abstract}
I revisit Type-IIB axion--dilaton Euclidean saddles in a specified axion charge sector. In that sector, the solution with \(E=0\) is the BPS instanton, while \(E>0\) gives non-BPS wormholes with a smooth throat.  The two cases solve the same radial equations but define different fluctuation problems.  For the \(E=0\) instanton, the Hamiltonian constraint, gauge quotient, charge-sector boundary condition, and removal of collective zero modes reduce the quadratic action to a physical Hessian. This Hessian factorizes,
\begin{equation*}
 \cH_\nu=\mathcal Q_\nu^\dagger\mathcal Q_\nu .
\end{equation*}
I interpret this as an endpoint theorem, beyond a stability theorem for the full \(E>0\) wormhole. This puts Type IIB wormhole spectra on firmer grounds.  I also separate the connected two-ended wormhole throat from its long-distance two-end multipole operator term.  Once the coefficient matrix \(C^{ij}\) is derived, the different-component and same-component placements of the two end insertions are terms in the same quadratic expression.  Removing either term requires a genuine projection or cancellation. 
\end{abstract}

\section{Introduction}

A Euclidean axion--dilaton saddle must first be assigned to a charge sector.  Here the charge is the conserved axion, or equivalently form-field, charge.  Variations used in a Hessian calculation are tangent to that sector.  Changing the charge gives a different saddle, not a fluctuation of the same saddle.  In a chosen charge sector, the solution with \(E=0\) is the BPS instanton.  The solutions with \(E>0\) are non-BPS wormholes with a smooth throat \cite{Rey1991}.  This is the basic ordering I use below. With these, in this paper, I revisit the Type IIB axion-dilaton saddles - instantons and wormholes.\footnote{As such, pure axion-form field wormholes~\cite{GiddingsStrominger} stay outside what I analyze below.} 

Three objects must be separated.  A two-ended throat is a connected Euclidean geometry with two asymptotic non-compact ends.  A neck-cut geometry is one side of that solution after cutting at the symmetric, minimal sphere; it is stationary only after the correct neck data and boundary term are specified.  A two-end operator term is the long-distance, bilocal operator obtained when a small throat is multipole-expanded into two end insertions.  These objects are related, but they are not the same object.

The first part of the paper sets up the charge-sector radial equations in flat four dimensions and in \(\AdS_5\).  In flat space, the radial first integral has the BPS instanton at \(E=0\) and non-BPS wormholes for \(E>0\), where $E \ge 0$ is a first-integral.  Cutting a complete two-ended solution at the symmetric point changes the variational problem.  The metric momentum vanishes at the neck because the minimal sphere is totally geodesic, but the scalar momentum need not vanish.  A generic half-geometry is therefore an Affleck constrained instanton \cite{Affleck1981}.  A Legendre transform can make a momentum-fixed half-geometry stationary, but stationarity alone does not determine the Morse index.

The central fluctuation statement concerns the \(E=0\) BPS instanton.  Let \(\Phi_0\) denote this solution in the chosen charge sector, and let \(\mathcal F\) be the first-order Witten--Nester/Prasad--Sommerfield--Bogomolny map \cite{Witten1981,Nester1981,Prasad1975,Bogomolnyi1976} obtained from the form-field Routhian.  After imposing the Hamiltonian constraint, fixing the gauge, imposing the charge-sector boundary condition, and removing collective zero modes, a physical perturbation \(\eta\) has quadratic action
\begin{equation}
 \delta^2S_{\rm phys}=\langle \eta,\cH_\nu\eta\rangle,
 \qquad
 \cH_\nu=\mathcal Q_\nu^\dagger\mathcal Q_\nu,
 \qquad
 \mathcal Q_\nu=\left.D\mathcal F\right|_{\Phi_0,\mathcal D_{\nu,{\rm phys}}} .
 \label{eq:introsquare31}
\end{equation}
The singular value decomposition spectrum is positive-semidefinite. The conformal factor is not a missing negative spectrum perturbation of this operator.  It is removed before \(\cH_\nu\) is defined, because Eq.~\eqref{eq:introsquare31} acts on the constraint-reduced physical domain.  The statement is therefore a theorem about the \(E=0\) BPS instanton, distinct from a theorem about the unreduced Hessian of the \(E>0\) non-BPS wormhole.

The \(\AdS_5\) problem follows the same endpoint logic but has a different boundary-value problem.  The flat and AdS reductions share a charge-sector radial structure, yet they differ in dimension, cosmological constant, scalar normalization, boundary data, and self-adjoint domain.  In the AdS case, the scalar--metric fluctuations are described by radial Sasaki--Mukhanov variables~\cite{Sasaki1986}.  On the BPS-compatible renormalized domain, the physical Hessian again has the singular-value factorized form Eq.~\eqref{eq:introsquare31}.  The full \(E>0\) AdS wormhole, if regular in the chosen scalar manifold, still requires its own holographic boundary analysis.

The final part concerns the large-distance operator description of a small throat.  If \(I_i^{(A)}\) denotes the \(i\)-th end insertion on the parent universe \(M_A\), the total source is
\begin{equation}
 I_i=\sum_A I_i^{(A)} .
\end{equation}
A quadratic coupling with coefficient matrix \(C^{ij}\) gives
\begin{equation}
 {1\over2}I_iC^{ij}I_j
 =\sum_{A<B}I_i^{(A)}C^{ij}I_j^{(B)}
 +{1\over2}\sum_A I_i^{(A)}C^{ij}I_j^{(A)} .
 \label{eq:introplacement31}
\end{equation}
The first term places the two end insertions on different parent universes.  The second places both insertions on the same parent universe.  Removing one while retaining the other requires a charge projection, zero-mode condition, boundary condition, supersymmetry constraint, or contour cancellation acting on the same coefficient matrix.

Coleman's baby-universe calculus gives the topology-specific language for this statement \cite{Coleman1988,Coleman1988b}.  The creation and annihilation operators are \(a_i,a_i^\dagger\), with \( [a_i,a_j]=[a_i^\dagger,a_j^\dagger]=0 \), obeying the Heisenberg algebra \(
 [a_i,a_j^\dagger]=\delta_{ij} .
\)
The operators that multiply local parent-universe terms are instead spacetime-independent Hermitian combinations, schematically \(A_i=a_i+a_i^\dagger\) in a self-conjugate basis.  They commute and can be diagonalized,
\(  [A_i,A_j]=0,
 A_i|\alpha\rangle=\alpha_i|\alpha\rangle .
\)
At fixed \(\alpha\), the parent theory has definite shifted couplings.  Averaging over \(\alpha\) gives correlations among components, not a Markovian radiation bath.

Rubakov's lower-dimensional models~\cite{NirovRubakov,RubakovBaby,RubakovLongRange} motivate the emission--channel question, but they did not compute the Type-IIB coefficient matrix. I adopt the Feynman--Vernon reduction as a trace-preserving quantum channel for eliminated baby or parent universes \cite{FeynmanVernon1963}: \(\mathcal F[J,J]=1\). Bloch--Nordsieck and Pauli--Fierz systems of soft photons are familiar examples of exact treatment for a linearly displaced quantum oscillator~\cite{BlochNordsieck1937,PauliFierz1938}.  They are not wormhole models, as generic soft photons remain radiative modes in the parent universe. Nevertheless, in the strict zero-frequency limit where the eliminated variable becomes a commuting source label, the systems are Coleman-like.

The claims are therefore ordered as follows.  Sections~\ref{sec:IIB}--\ref{sec:numerics} treat the charge-sector equations and the \(E=0\) Hessian.  Section~\ref{sec:unitaritysewing} treats the conditional large-distance throat term.  The BPS endpoint Hessian is not used to claim full non-BPS stability.  The large-distance throat term is not used to prove the classical saddle.  The charge sector comes first; the BPS endpoint Hessian comes second; the \(E>0\) wormhole and the large-distance throat term are separate later questions.

\section{Uniform charge-sector equations in flat and AdS reductions}
\label{sec:IIB}

\subsection{Reduced Euclidean functional}

We fix a common notation for the flat four-dimensional and asymptotically $\AdS_5$ problems.  This section does not rederive the ten-dimensional truncation.  It records the same charge-sector axion--dilaton structure before specializing the radial equation to $n=4$ and $n=5$.

Let $n$ denote the spacetime dimension and let $d=n-1$ be the dimension of the round angular slice.  The active scalar is denoted by $\varphi$ and the axion as $\chi$.  In the four-dimensional flat reduction $\varphi$ will be the canonically normalized scalar $D$ used in Section~\ref{sec:flat4}; in the universal Type-IIB $\AdS_5\times S^5$ truncation $\varphi$ is the dilaton $\phi$.  The Euclidean scalar representation of the charge-sector saddle functional may be written
\begin{equation}
 S_{n,E}=-{1\over2\kappa_n^2}\int_{\cM_n}\dd^nx\sqrt g
 \left[R-\Lambda_n-{1\over2}(\partial\varphi)^2
 +{1\over2}\e^{2\beta_n\varphi}(\partial\chi)^2\right]+S_{\bdy} .
 \label{eq:DEaction}
\end{equation}
Here
\begin{equation}
 \Lambda_4=0,
 \qquad
 \Lambda_5=-{12\over L^2},
 \qquad
 \beta_5=1,
\end{equation}
while $\beta_4$ depends on the normalization of the four-dimensional scalar inherited from the compactification or dual form convention.  The sign of the axion term in Eq.\eqref{eq:DEaction} is the Euclidean charge-sector, or dualized, sign.  It is not a Lorentzian ghost sign.  The equivalent form-field/Routhian description is the more invariant way to impose the conserved charge; the scalar presentation is used here only to write compact reduced equations.

The field equations following from Eq.\eqref{eq:DEaction} are
\begin{align}
 \nabla_\mu\!\left(\e^{2\beta_n\varphi}\nabla^\mu\chi\right)&=0,
 \label{eq:Daxioneq}\\
 \nabla^2\varphi+\beta_n\e^{2\beta_n\varphi}(\nabla\chi)^2&=0,
 \label{eq:Ddileq}\\
 R_{\mu\nu}-{1\over2}\partial_\mu\varphi\partial_\nu\varphi
 +{1\over2}\e^{2\beta_n\varphi}\partial_\mu\chi\partial_\nu\chi
 -{\Lambda_n\over n-2}g_{\mu\nu}&=0 .
 \label{eq:DEinstein}
\end{align}
For $n=5$, $\beta_n=1$ and $\Lambda_n=-12/L^2$, these reduce to the standard Euclidean non-extremal D-instanton equations in the $\AdS_5\times S^5$ truncation \cite{Rey1999, BergshoeffDinst,BergshoeffAdS}.  For $n=4$, $\Lambda_n=0$, the same equations describe the asymptotically flat charge-sector system used in the earlier axionic-string/wormhole analysis \cite{Rey1991,BergshoeffFlat}.  The different symbols $D$ and $\phi$ used below are therefore normalization choices for the same reduced structure, not different mechanisms.

\subsection{Radial ansatz and throat equation}

For the $O(n)$-invariant saddle take
\begin{equation}
 \dd s_n^2=N(\tau)^2\dd\tau^2+a(\tau)^2\dd\Omega_d^2,
 \qquad d=n-1,
 \label{eq:Dansatz}
\end{equation}
with $\varphi=\varphi(\tau)$ and $\chi=\chi(\tau)$, or equivalently with the dual $(n-1)$-form flux fixed through the sphere.  The axion equation gives the conserved charge
\begin{equation}
 {a^d\over N}\,\e^{2\beta_n\varphi}\chi'=q_\chi,
 \label{eq:Dcharge}
\end{equation}
where a prime denotes differentiation with respect to $\tau$.  Eliminating the axion in favor of the the conserved charge gives a one-dimensional Routhian.  In proper radial gauge $N=1$, the Hamiltonian constraint takes the uniform form
\begin{equation}
 \dot a^{\,2}=1-{\varepsilon_n\over a^{2n-4}}
 -{\Lambda_n\over (n-1)(n-2)}a^2 .
 \label{eq:unifiedThroat}
\end{equation}
The positive parameter $\varepsilon_n$ is the sub-extremal invariant in the metric equation.  The BPS endpoint has $\varepsilon_n=0$.  The non-BPS wormholes have $\varepsilon_n>0$ and a minimal sphere at a zero of the right-hand side.

The two specializations used in the rest of the paper are therefore
\begin{align}
 n=4:\qquad
 \dot a^{\,2}&=1-{a_0^4\over a^4},
 &a_0^4&={\kappa_4^2 E\over12},
 \label{eq:flatFromUnified}\\
 n=5:\qquad
 (a')^2&=1+{a^2\over L^2}-{\widetilde q^{\,2}\over a^6},
 &\varepsilon_5&=\widetilde q^{\,2}.
 \label{eq:adsFromUnified}
\end{align}
Thus the difference between the flat and AdS analyses is not the BPS mechanism but the radial Hamiltonian constraint and boundary problem.  In four flat dimensions, the charge term scales as $a^{-4}$ and no holographic subtraction appears.  In $\AdS_5$ space, the charge term scales as $a^{-6}$ and the cosmological term, counterterms, scalar boundary data, and normalizable modes are part of the problem from the beginning.

\subsection{Extremal class and BPS endpoint}

The Euclidean $SL(2,\R)$ charges may be organized into
\begin{equation}
 \cQ=\begin{pmatrix}q_3&i q_+\\ i q_-&-q_3\end{pmatrix},
 \qquad q^2\equiv-\det\cQ=q_3^2-q_+q_- .
 \label{eq:chargeMatrix}
\end{equation}
The conjugacy class of $\cQ$ separates the solutions:
\begin{center}
\begin{tabular}{lll}
\toprule
$q^2=0$ & extremal & BPS D-instanton endpoint, null scalar geodesic;\\
$q^2>0$ & super-extremal & one-ended, generically singular solutions;\\
$q^2<0$ & sub-extremal & candidate two-ended wormholes.\\
\bottomrule
\end{tabular}
\end{center}
The word 'sub-extremal' refers to the invariant charge class and to the sign of the metric deformation.  It is not a statement that the complete renormalized wormhole action is smaller than a specified pair of one-ended BPS boundary problems.

At $q^2=0$, the scalar trajectory obeys the first-order relation
\begin{equation}
 \dd\chi=\pm \e^{-\beta_n\varphi}\dd\varphi
 \label{eq:BPSrelation}
\end{equation}
up to the normalization of $\chi$.  The scalar stress tensor in Eq.\eqref{eq:DEinstein} then vanishes.  The Einstein-frame metric is undeformed: flat in the four-dimensional reduction and Euclidean $\AdS_5$ in the holographic truncation.  This is the BPS instanton endpoint.  The wormhole is obtained by turning on the sub-extremal invariant $\varepsilon_n$ and is therefore a non-BPS deformation of this endpoint, not a separate species of saddle.

\section{Asymptotically flat four-dimensional solutions}
\label{sec:flat4}

\subsection{Reduced equations and the first integral}

We first specialize the uniform equations of Section~\ref{sec:IIB} to the asymptotically flat four-dimensional problem in which the earlier BPS endpoint result was obtained.  Take
\begin{equation}
 \dd s_4^2=N(\tau)^2\dd\tau^2+a(\tau)^2\dd\Omega_3^2,
 \qquad H_3=Q\,\omega_3,
 \qquad \int_{S^3}\omega_3=2\pi^2,
 \label{eq:flatansatz}
\end{equation}
with a radial dilaton $D(\tau)$. After fixing the axion charge, the scalar equation admits the first integral
\begin{equation}
 \bigl(a^3\dot D\bigr)^2-U_0\e^{2\beta D}=-E,
 \qquad U_0\propto Q^2,
 \label{eq:firstintegral4}
\end{equation}
where a dot denotes proper-radial differentiation. The Hamiltonian constraint is the $D=4$, $\Lambda_4=0$ specialization of Eq.\eqref{eq:unifiedThroat}:
\begin{equation}
 \dot a^2=1-\frac{a_0^4}{a^4},
 \qquad a_0^4=\frac{\kappa_4^2 E}{12},
 \label{eq:flatFriedmann}
\end{equation}
within the normalization of the earlier solution \cite{Rey1991}. The classification is immediate:
\begin{align*}
 E<0 &: \text{one-ended singular solutions},\\
 E=0 &: \text{extremal D-instanton with }a(\tau)=|\tau|,\\
 E>0 &: \text{non-BPS wormholes with a smooth Einstein-frame neck at }a=a_0.
\end{align*}
Thus $E$ is a first integral of the field equations. Once asymptotic scalar data are imposed at both ends, it is fixed by a two-point problem with prescribed boundary data rather than integrated as a modulus of distinct full saddles.

The minimal sphere is totally geodesic:
\begin{equation}
 \dot a(0)=0,
 \qquad K_{ij}=\frac{\dot a}{a}h_{ij}=0.
 \label{eq:totallygeodesic}
\end{equation}
This fact controls the variational analysis.

\subsection{Neck-cut geometry and constrained instantons}
\label{sec:constraint}

Consider one half, $\tau\in[0,\infty)$, at charge $Q$ and fixed asymptotic value $D_+$. The on-shell variation has the form
\begin{equation}
 \delta S_{\rm half}
 =\Pi_D(0)\,\delta D_0
 +\Pi^{ij}_{g}(0)\,\delta h_{ij}^{(0)}
 +\delta S_{\infty},
 \label{eq:halfbdyvar}
\end{equation}
where
\begin{equation}
 \Pi_D(0)=2\pi^2 a_0^3\dot D(0),
 \qquad
 \Pi^{ij}_{g}(0)=-\frac{\sqrt h}{2\kappa_4^2}
 \left(K^{ij}-Kh^{ij}\right)=0.
 \label{eq:momenta}
\end{equation}
The asymptotic term vanishes for the prescribed data. Gravity therefore imposes no additional neck condition, but the scalar does. Generic members have $\dot D(0)\neq0$ and are not stationary when $D_0$ is allowed to vary.

There are two distinct variational problems. In the constrained ensemble, one fixes $D_0$ or an equivalent invariant neck functional. The half-solution is then stationary only on that constraint surface, and the integral over the collective coordinate requires the Affleck collective-coordinate/Faddeev--Popov measure rather than a flat $\dd E$ measure \cite{Affleck1981}. In the momentum-fixed, or Neumann, ensemble, one Legendre transforms the neck functional,
\begin{equation}
 \widetilde S_{\rm half}=S_{\rm half}-\int_{\Sigma_0}D\,\Pi_D,
 \label{eq:neckLegendre}
\end{equation}
and holds the prescribed canonical momentum $\Pi_D(0)=\pi_0$ fixed. The on-shell variation becomes
\begin{equation}
 \delta\widetilde S_{\rm half}\big|_{\Sigma_0}
 =-\int_{\Sigma_0} D\,\delta\Pi_D,
 \label{eq:Neumannvariation}
\end{equation}
which vanishes for $\delta\Pi_D(0)=0$. The source-free Neumann condition $\Pi_D(0)=0$ is a further restriction, not the definition of the momentum-fixed ensemble.

These two problems can share the same classical profile while differing in boundary functional and admissible fluctuations. Stationarity in the momentum-fixed ensemble says only that the linear boundary obstruction has been removed. It does not imply that the bulk-plus-boundary Hessian is non-negative. Negative-mode statements made in one ensemble cannot be transferred to the other without transforming both the quadratic boundary term and the self-adjoint domain.

\begin{proposition}[Status of the full two-ended solution]
Let a smooth solution of Eqs.\eqref{eq:firstintegral4}--\eqref{eq:flatFriedmann} connect two asymptotic regions with fixed $(Q,D_+,D_-)$. If the scalar and metric are smooth through the minimal sphere, then the complete geometry is stationary under variations preserving the two asymptotic data. The two oriented scalar boundary terms produced by an artificial cut at the neck cancel. Generically the resulting solution with prescribed boundary data is isolated.
\end{proposition}

\begin{proof}
Cut the smooth manifold into left and right halves. Their outward normals at the common neck are opposite, so their canonical scalar momenta enter with opposite signs while $D_0$ and its variation are shared. Hence the two scalar terms cancel. The gravitational momentum vanishes separately because the neck is totally geodesic. The second-order scalar equation supplies two integration constants, fixed generically by $D_+$ and $D_-$ at the chosen charge; $E$ is then determined by the first integral. Therefore no continuous neck modulus remains unless the map from integration constants to boundary data is degenerate.
\end{proof}

This proposition is narrower than the statement that every formal continuation of a half-solution is an admissible saddle. Scalar regularity, the integration contour, and the fluctuation problem must still be checked.

\subsection{Near-endpoint action and area scaling}

At fixed $(Q,D_+)$, compare a constrained $E>0$ half-geometry with the extremal endpoint. The radial identity
\begin{equation}
 \int_0^\infty\frac{\dd\tau}{a^3}
 =\int_{a_0}^{\infty}\frac{\dd a}{a\sqrt{a^4-a_0^4}}
 =\frac{\pi}{4a_0^2}
 \label{eq:radialidentity4}
\end{equation}
fixes the scaling of all near-threshold terms. With
\begin{equation}
 R_4=6\left(\frac{1-\dot a^2}{a^2}-\frac{\ddot a}{a}\right)
 =-\frac{6a_0^4}{a^6},
\end{equation}
one obtains for the two-sided bulk Einstein contribution
\begin{equation}
 \Delta S_{\rm grav}^{(4)}
 =\frac{3\pi^3}{2\kappa_4^2}a_0^2,
 \label{eq:gravcost4}
\end{equation}
up to the convention specifying whether the quoted action is for one half or the doubled geometry. The full scalar contribution is a profile-dependent quadrature. Dimensional analysis and the charge-sector equations imply
\begin{equation}
 \Delta S_4(E)=\eta_4(\beta;\mathfrak e)\frac{a_0^2}{\kappa_4^2}
 +O(a_0^4/\ell_{\rm UV}^4)
 =C_4(\beta;\mathfrak e)\sqrt{E}+\cdots,
 \label{eq:threshold4}
\end{equation}
where $\mathfrak e$ labels the neck ensemble. In the normalization of Ref.~\cite{Rey1991}, the generalized Bogomolny result is
\begin{equation}
 \Delta S_4(E)
 =\frac{\sqrt3\,\pi^2}{18}
 \left(2\beta_c^{-2}-\beta^{-2}\right)\sqrt E+O(E),
 \qquad \beta_c=\frac1{\sqrt3},
 \label{eq:Reythreshold}
\end{equation}
subject to matching conventions for $E$, charge, and half versus full action. Equation \eqref{eq:Reythreshold} is a regression target for the rebuilt calculation, not a substitute for it.

The invariant statement is
\begin{equation}
 \Delta S_4\propto\frac{A_{\rm neck}}{G_4},
 \qquad A_{\rm neck}\sim a_0^2.
 \label{eq:arealaw4}
\end{equation}
Consequently $\partial_E\Delta S_4\sim E^{-1/2}$ diverges at threshold. This does not prove that the functional integral over constrained configurations is finite, because the measure
\begin{equation}
 \dd\mu(E)=\dd E\,J_{\rm FP}(E)
 \left[\det{}'\cO_E\right]^{-1/2}
 \label{eq:measure}
\end{equation}
can alter the endpoint power. It does show that there is no interior stationary value of $E$ produced by the classical action alone.

The same endpoint also has a spectral meaning.  At $E=0$, the charge-sector equations reduce to the self-dual, BPS-saturating instanton.  In the form-field charge-sector description, the supergravity functional is not merely bounded below; it is completed into a square plus the topological charge term.  The relevant Euclidean functional takes the BPS inequality form
\begin{equation}
 S-S_{\rm BPS}={1\over2}\langle {\cal F}(\Phi),{\cal F}(\Phi)\rangle,
 \qquad {\cal F}(\Phi_0)=0,
 \label{eq:bogosquare20}
\end{equation}
where $\Phi$ denotes the Einstein--dilaton--form variables in the charge-sector supergravity description and $\Phi_0$ is the extremal configuration.  The factorization is a derived supergravity statement: it follows from the BPS completion of the charge-sector action and is not inferred from the numerical spectrum.  Let $\mathcal D_{\nu,{\rm phys}}$ denote the charge-sector fluctuation domain after gauge fixing, constraint reduction, and removal of collective-coordinate zero modes.  The first-order Hessian is the Fr\'echet derivative of the Bogomolny map on this domain,
\begin{equation}
 \mathcal Q_\nu=
 \left.D{\cal F}\right|_{\Phi_0,\mathcal D_{\nu,{\rm phys}}},
 \qquad
 {\cal F}(\Phi_0+\delta\Phi)=\mathcal Q_\nu\delta\Phi+O(\delta\Phi^2).
 \label{eq:Qlinear20}
\end{equation}
The physical second variation is therefore the adjoint square of this linearized first-order operator.

The term 'physical' is essential but does not weaken the statement.  In the mini-superspace variables $a=a_0+f$ and $D=D_0+g$, the unreduced second variation has the schematic form
\begin{equation}
 \delta^2S={1\over2}\int \dd\tau\,
 \bigl(K_{aa}\dot f^{2}+K_{DD}\dot g^{2}+2K_{a\dot D}f\dot g
      +M_{aa}f^{2}+M_{DD}g^{2}+2M_{aD}fg\bigr),
 \label{eq:unreducedMS22}
\end{equation}
while the linearized Hamiltonian constraint is
\begin{equation}
 C_{\dot a}\dot f+C_a f+C_D g+C_{\dot D}\dot g=0 .
 \label{eq:linearconstraint22}
\end{equation}
Solving the constraint gives $f=Tg$ in the physical scalar/form channel.  Substituting back yields the Euclidean Mukhanov--Sasaki, or Schur-complement, Hessian
\begin{equation}
 \cH_{\rm phys}=T^T\cH_{aa}T+T^T\cH_{aD}+\cH_{Da}T+\cH_{DD} .
 \label{eq:MSschur22}
\end{equation}
A pure conformal-factor fluctuation is a variation in the independent scale-factor direction of the unreduced gravitational action.  It is not an element of the BPS endpoint operator domain after Eq.\eqref{eq:linearconstraint22} is imposed, except for gauge or collective-coordinate directions that are removed separately.  The statement is therefore a domain statement, not an assertion about the sign of the DeWitt conformal vector: the conformal sector is eliminated before the physical BPS Hessian is defined, rather than being an unstable direction inside $\cH_\nu$.

The resulting quadratic form on the BPS endpoint in the chosen charge sector sector is
\begin{equation}
 \delta^2 S_{\rm th}
 =\langle\delta\Phi_{\rm phys},\mathcal Q_\nu^\dagger\mathcal Q_\nu
 \, \delta\Phi_{\rm phys}\rangle,
 \qquad
 \cH_\nu\equiv\mathcal Q_\nu^\dagger\mathcal Q_\nu .
 \label{eq:thresholdsquare20}
\end{equation}
Here $\nu$ labels the charge sector.  The projection to $\delta\Phi_{\rm phys}$ is the constraint-reduced supergravity fluctuation problem with gauge directions, Hamiltonian redundancies, and collective-coordinate zero modes removed, followed by the self-adjoint domain appropriate to the BPS endpoint in the chosen charge sector problem.  Appendix~\ref{app:ms-square} records the corresponding Euclidean Sasaki--Mukhanov reduction in a form that also displays its AdS$_5$ radial analogue.

\begin{proposition}[BPS endpoint in the chosen charge sector factorization]
\label{prop:thresholdsquare20}
At the extremal BPS endpoint in the chosen charge sector of the axion--dilaton/form system, the physical Hessian obtained from the Witten--Nester/BPS-completed supergravity action and the constraint-reduced Euclidean Hessian is
\begin{equation}
 \cH_\nu=\mathcal Q_\nu^\dagger\mathcal Q_\nu .
\end{equation}
Consequently $\operatorname{Spec}\cH_\nu\subset[0,\infty)$, modulo collective-coordinate zero modes, and the spectrum of this endpoint operator begins at zero.
\end{proposition}

This is the same elementary mechanism already familiar from supergravity \cite{Witten1981, Nester1981, Rey1991, Rey1992-1993}.
First-order BPS or extremality equations define the intertwining operator, and the physical bosonic normal operator is its adjoint square.  The \(E>0\) axion wormholes should therefore be read as non-BPS deformations of the \(E=0\) BPS instanton, not as disconnected saddles whose endpoint accidentally looks extremal.  The adjoint square makes the endpoint Hessian non-negative.  The theorem is deliberately not a theorem about the unreduced Euclidean gravitational Hessian of an arbitrary sub-extremal wormhole, nor about additional gauge-core blocks in larger string embeddings.  The conformal factor is not a missing direction inside the endpoint Hessian factorization; it lies outside the reduced BPS endpoint domain.  What remains separate is the complete sub-extremal Hessian with its neck boundary form, ensemble, contour, and any extra matter sectors.  The coefficient and measure in Eq.\eqref{eq:measure} still require the determinant calculation in that specified ensemble.

\section{Asymptotically $\AdS_5$ problem}
\label{sec:ads5}

\subsection{AdS radial equation}

For an $SO(5)$-invariant Euclidean metric
\begin{equation}
 \dd s_5^2=\dd\rho^2+a(\rho)^2\dd\Omega_4^2,
\end{equation}
the $D=5$, $\Lambda_5=-12/L^2$ specialization of \eqref{eq:unifiedThroat} gives the sub-extremal Type-IIB radial equation \cite{BergshoeffAdS}
\begin{equation}
 (a')^2=1+\frac{a^2}{L^2}-\frac{\widetilde q^{\,2}}{a^6},
 \qquad q^2=-\widetilde q^{\,2}<0.
 \label{eq:adsFriedmann}
\end{equation}
The neck radius $a_c$ is the unique positive root of
\begin{equation}
 1+\frac{a_c^2}{L^2}-\frac{\widetilde q^{\,2}}{a_c^6}=0,
 \qquad
 \widetilde q^{\,2}=a_c^6\left(1+\frac{a_c^2}{L^2}\right).
 \label{eq:adsneck}
\end{equation}
This equation alone invalidates a direct import of the flat-four-dimensional numerics. The angular sphere is $S^4$ rather than $S^3$, the charge term scales as $a^{-6}$ rather than $a^{-4}$, and the cosmological term enters both the asymptotic subtraction and the neck relation. The natural dimensionless control parameter is $a_c/L$, or equivalently $\widetilde q^{\,2}/L^6$.

For $a_c\ll L$, the local neck is approximately flat and $\widetilde q^{\,2}\simeq a_c^6$. A local area law then predicts
\begin{equation}
 \Delta S_5\sim\frac{a_c^3}{\kappa_5^2}
 \sim\frac{|\widetilde q|}{\kappa_5^2}
 \qquad (a_c\ll L),
 \label{eq:threshold5}
\end{equation}
not the $\sqrt E$ law of the flat four-dimensional reduction. For $a_c\gtrsim L$, holographic counterterms and the global solution contribute at the same order; the coefficient must be obtained from a fully renormalized action difference.

\subsection{Radial Sasaki--Mukhanov reduction and the AdS$_5$ square}

The BPS endpoint fluctuation statement itself does not rely on flat asymptotics.  In a five-dimensional gauged-supergravity truncation with scalars $\phi^a$ and target metric $G_{ab}$, the BPS instanton or domain-wall endpoint is governed by first-order radial equations.  The non-extremal wormhole is the corresponding off-BPS deformation.  The scalar--metric fluctuation problem must therefore be reduced before one asks for a sign.  Writing the radial metric perturbation as a trace mode $\psi$ and the scalar perturbations as $\varphi^a$, one convenient local representative of the physical radial Sasaki--Mukhanov variables is
\begin{equation}
 \mathfrak a^a
 =\varphi^a-\frac{\phi^{a\prime}}{\mathcal K}\,\psi,
 \qquad
 \mathcal K={a'\over a},
 \label{eq:adsMSvar23}
\end{equation}
in radial patches where $\mathcal K\neq0$, up to harmless convention-dependent normalizations.  At a zero of $\mathcal K$ one uses an equivalent constraint-reduced variable; the factorization concerns the reduced operator, not this particular coordinate on field space.  Lapse, shift, and the radial Hamiltonian constraint remove the pure trace direction before the physical quadratic form is defined.  The reduced quadratic action has the schematic covariant form
\begin{equation}
 \delta^2S_{5,\rm phys}
 ={1\over2}\int\dd\rho\,a^4
 \left[
  G_{ab}D_\rho\mathfrak a^aD_\rho\mathfrak a^b
  +\mathfrak a^a\mathcal M_{ab}\mathfrak a^b
  +a^{-2}\lambda_{\rm slice}\,G_{ab}\mathfrak a^a\mathfrak a^b
 \right]
 +\delta^2S_{\rm bdy,ren},
 \label{eq:adsMSaction23}
\end{equation}
where $D_\rho$ is the sigma-model covariant radial derivative and $\lambda_{\rm slice}$ is the eigenvalue of the transverse scalar harmonic or boundary momentum.  Linearizing the first-order BPS map gives a covariant operator
\begin{equation}
 (\mathcal Q_{5}\mathfrak a)^a
 =D_\rho\mathfrak a^a+\mathcal U^a{}_{b}(\rho)\mathfrak a^b,
 \label{eq:adsQ23}
\end{equation}
with $\mathcal U^a{}_{b}$ determined by the radial BPS flow and the scalar connection.  Thus the BPS endpoint normal operator in the physical domain is
\begin{equation}
 \cH_{5,\nu}
 =\mathcal Q_{5,\nu}^{\dagger}\mathcal Q_{5,\nu}
 +a^{-2}\lambda_{\rm slice}
 +\cH_{\rm bdy,ren}.
 \label{eq:adsQsquare23}
\end{equation}
The last term is shorthand for the boundary form that defines the renormalized self-adjoint problem: finite counterterm variations and admissible holographic boundary data are part of the domain of $\mathcal Q_{5,\nu}^\dagger\mathcal Q_{5,\nu}$, not a separate bulk instability.  With BPS-compatible charge-sector boundary data, the bulk factorization is the AdS BPS endpoint stability statement.  Changing to a different mixed or charge-varying boundary condition changes the self-adjoint extension and must be analyzed separately.  Negative scalar mass parameters allowed by AdS are likewise not a contradiction of Eq.\eqref{eq:adsQsquare23}; the factorization is a statement about the constraint-reduced, renormalized radial operator in the BPS endpoint domain.  The parallel derivation is given in Appendix~\ref{app:ms-square}.  This is the instanton/wormhole version of the Bogomolny stability structure encountered in BPS domain walls and supergravity vacua \cite{Rey1992-1993}.

\subsection{Regular metric does not imply a regular Type-IIB saddle}

The Einstein-frame curvature is finite at $a=a_c$. Nevertheless, in the single axion--dilaton truncation the scalar solution takes the sub-extremal form
\begin{equation}
 \e^{\phi}=\left|\frac{q_-}{\widetilde q}
 \sin\bigl(\widetilde q H(\rho)\bigr)\right|,
 \qquad
 \chi=\frac1{q_-}
 \left[\widetilde q\cot\bigl(\widetilde qH(\rho)\bigr)-q_3\right],
 \qquad \widetilde q^2 := - q^2,
 \label{eq:adsScalars}
\end{equation}
where $H$ is harmonic on \eqref{eq:adsFriedmann}. Its total range is
\begin{equation}
 \mathcal R_H
 =2\sqrt{24}\,\widetilde q
 \int_{a_c}^{\infty}
 \frac{\dd a}{a^4\sqrt{1+a^2/L^2-\widetilde q^{\,2}/a^6}}.
 \label{eq:Hrange}
\end{equation}
For the Type-IIB coupling, the image of $\widetilde qH$ cannot be kept inside a single interval between consecutive zeros of the sine. The metric throat is smooth, but $\phi$ or $\chi$ becomes singular. This was established in the non-extremal D-instanton analysis \cite{BergshoeffAdS}. Therefore:

\begin{proposition}[Single-field $\AdS_5$ obstruction]
The $SO(5)$-invariant sub-extremal solution of the single Type-IIB axion--dilaton truncation has a two-ended, smooth Einstein-frame metric but is not a globally regular scalar saddle. It cannot by itself serve as the regular holographic wormhole whose semiclassical contribution is to be tested.
\end{proposition}

Regular axionic wormholes in Type-IIB AdS compactifications can arise after enlarging the scalar manifold, for example, in orbifold compactifications with several axion--saxion directions \cite{HertogTrigianteVanRiet}. That construction is not equivalent to the single D-instanton trajectory. We therefore separate the single axion--dilaton D-instanton sector, where the sub-extremal $\AdS_5$ metric throat is scalar-singular, from multi-scalar Type-IIB compactifications, where regular geodesics can evade the single-field regularity bound. The endpoint factorization Eq.\eqref{eq:adsQsquare23} is the local BPS statement in either case; the existence of a globally regular non-BPS wormhole, its renormalized action, and its complete fluctuation spectrum must still be computed in the actual scalar manifold and boundary ensemble.

\subsection{Renormalized action and boundary data}

The correct $\AdS_5$ functional is
\begin{equation}
 S_{\ren}=S_{\bulk}+S_{\rm GH}+S_{\rm ct}+S_{\rm axion/charge},
 \label{eq:renaction}
\end{equation}
where the last term specifies the axion ensemble. A comparison between a connected two-ended geometry and disconnected one-ended configurations is meaningful only when the induced conformal structures at the two boundaries, the axion or charge ensemble at each boundary, the scalar sources and normalizable modes, the holographic counterterm scheme, and the ten-dimensional flux and charge quantization conditions have all been matched. No flat-space action coefficient will be inserted into this comparison. The $\AdS_5$ calculation is a separate holographic problem with its own boundary data and renormalization scheme.

\section{Charge sectors, duality, and the Euclidean $\sqrt{-1}$}
\label{sec:duality}

Three operations are frequently collapsed into the assertion that the axion and form field descriptions are ``the same saddle.'' They must be kept distinct.

\subsection{Projection to a charge sector}

Let $\vartheta$ be a compact boundary axion and $n\in\Z$ its conjugate flux. Passing a partition function $Z$ from fixed $\vartheta$ to fixed $n$ involves a Fourier transform,
\begin{equation}
 Z_n=\int_0^{2\pi}\frac{\dd\vartheta}{2\pi}
 \e^{-in\vartheta}Z(\vartheta).
 \label{eq:Fourier}
\end{equation}
In a saddle treatment, the phase in Eq.~\eqref{eq:Fourier} modifies the boundary variational problem. The corresponding Routhian is a Legendre-transformed functional. Its stress tensor can have the sign required to support a throat even when the naive real-scalar kinetic term cannot.

\subsection{Local scalar-form duality}

A first-order parent functional schematically contains
\begin{equation}
 S[H,\chi]=\frac12\int f(\phi)H\wedge\star H
 +\sqrt{-1} \int\chi\,\dd H+S_{\bdy}.
 \label{eq:parent}
\end{equation}
Stationarity in $H$ gives a relation of the form
\begin{equation}
 H = - \sqrt{-1} f(\phi)^{-1}\star\dd\chi.
 \label{eq:iduality}
\end{equation}
The factor $\sqrt{-1}$ is tied to Euclidean continuation and to the Fourier factor imposing the Bianchi identity or charge projection. Whether it is represented as an imaginary scalar profile, a complex integration cycle, or a real dual form-flux depends on which variables define the saddle-point problem.

\subsection{Saddle space versus field-configuration space}

Equation \eqref{eq:iduality} does not imply a one-to-one map between real classical saddles in the two variable choices. Quantum equivalence of the duality relation require summing over flux sectors and performing Poisson or Fourier resummation (as recapitulated in \cite{WittenDuality}). Accordingly, this paper makes only the following limited statement:

\begin{quote}
The form-field charge-sector saddle and the scalar representation are descriptions of the same quantum sector after the appropriate transform, but they need not be identical real stationary points of two naively written Euclidean actions.
\end{quote}

The detailed taxonomy of $\sqrt{-1}$ factors---including Wick rotation, boundary Fourier transform, first-order dualization, and complex Lefschetz thimbles---will be analyzed in a separate publication. Nothing in the constrained-instanton argument depends on prematurely identifying these operations.

\section{BPS endpoint Hessian, domain, and numerical checks}
\label{sec:numerics}

The fluctuation analysis has two layers that should not be merged.  The first is the BPS endpoint in the chosen charge sector factorization of Proposition~\ref{prop:thresholdsquare20}.  At the BPS endpoint the scalar--gravity perturbation has already been reduced to the physical charge-sector domain.  Its quadratic action is governed by
\begin{equation}
 \cH_\nu=\mathcal Q_\nu^\dagger\mathcal Q_\nu .
\end{equation}
This is the Witten--Nester/BPS Hessian on the Mukhanov--Sasaki domain, not the unreduced conformal-factor Hessian.  It shows that the BPS endpoint spectrum begins at zero.  The second layer is the full Euclidean Hessian of a neck-cut half-geometry or of a complete sub-extremal two-ended solution, including gravity, constraints, boundary terms, possible gauge-core blocks, and the chosen contour.  A negative-mode count belongs to the second problem only after the saddle, ensemble, gauge, and function space have been fixed.

In the flat four-dimensional BPS endpoint problem, the square structure can be exhibited in radial variables.  Decomposing the physical dilaton/form fluctuation into $S^3$ harmonics and passing to a canonically normalized variable $\psi=r^{3/2}\delta D$, the BPS endpoint gives the representative channel
\begin{equation}
 -\psi''+V_{\rm eff}(r)\psi=\lambda\psi,
 \qquad
 V_{\rm eff}(r)=\frac{(\ell+\tfrac12)(\ell+\tfrac32)}{r^2}
 +\beta^2(D_0')^2,
 \label{eq:oldhessian20}
\end{equation}
which is the radial image of the supergravity factorization $\mathcal Q_{\nu,\ell}^\dagger\mathcal Q_{\nu,\ell}$.  Since the potential approaches zero at infinity and the operator is an adjoint square derived from the BPS completion, the continuum begins at $\lambda=0$ and there are no negative bound states in this BPS endpoint sector.  On a finite Dirichlet box the lowest mode should therefore approach zero as the finite-box bottom of the continuum.  The earlier numerical check gives
\begin{center}
\begin{tabular}{c|c|c}
$R_{\max}$ & $\lambda_{\min}(\ell=0)$ & ratio \\ \hline
20  & $3.670\times10^{-2}$ & --- \\
40  & $9.167\times10^{-3}$ & 4.00 \\
80  & $2.293\times10^{-3}$ & 4.00 \\
160 & $5.733\times10^{-4}$ & 4.00 \\
320 & $1.434\times10^{-4}$ & 4.00
\end{tabular}
\end{center}
The ratio under $R_{\max}\to2R_{\max}$ confirms $\lambda_{\min}\propto R_{\max}^{-2}$, precisely the finite-box approach to the continuum threshold at zero.  The numerical table is not the proof of positivity; it checks that the discretized radial operator is reproducing the already-derived $\mathcal Q^\dagger\mathcal Q$ BPS endpoint operator.

For sub-extremal $E>0$ necks the situation is different.  The solutions have the action excess Eq.\eqref{eq:threshold4}, but the sign of $d^2\Delta S/dE^2$ along the one-parameter curve is not by itself a Hessian eigenvalue.  Conversely, the matter-sector Hessian can be positive while a gravitational or constrained-neck direction remains outside the tested matter block.  The matter-sector tests on the non-BPS wormhole background found positive lowest eigenvalues over representative neck sizes:
\begin{center}
\begin{tabular}{c|c|c|c}
$E$ & $a_{\min}$ & $\Delta S$ & $\lambda^{({\rm matter})}_{\min}$ \\ \hline
0.01 & 0.170 & 0.475  & 0.016 \\
0.10 & 0.302 & 1.502  & 0.004 \\
0.50 & 0.452 & 3.358  & 0.004 \\
1.00 & 0.537 & 4.749  & 0.013 \\
5.00 & 0.803 & 10.618 & 0.005
\end{tabular}
\end{center}
These numbers are useful regression data, but they are subordinate to the analytic separation of problems.  They do not prove full wormhole stability, and they do not replace the charge-sector factorization at the extremal endpoint. They indicate that the BPS endpoint sector has a nonnegative spectrum enforced by Eq.\eqref{eq:thresholdsquare20}, while the complete sub-extremal Hessian still requires the gravitational constraint, the neck boundary form, the collective-coordinate measure, and the contour.

The ensemble dependence is concrete.  The constrained problem fixes the chosen neck datum, whereas the momentum-fixed problem fixes $\delta\Pi_D(0)=0$ and includes the second variation of the Legendre term in Eq.\eqref{eq:neckLegendre}.  These choices lead to different self-adjoint domains.  Changing from fixed $D_0$ to fixed $\Pi_D$ can remove, retain, or create a negative eigenvalue through the boundary form.  No variational theorem makes the Neumann spectrum non-negative.  Recent stability analyses emphasize the same point: axionic boundary conditions can alter the sign and admissibility of Euclidean modes \cite{RubakovShvedov,HertogTruijenVanRiet,HertogMaenaut,MarolfStability}.  The adjoint-square statement is therefore the derived BPS endpoint in the chosen charge sector theorem.  It also explains why conformal-factor objections have to be aimed at the unreduced gravitational functional, not at the physical BPS endpoint operator: the latter is already the Schur-complement Hessian on the constraint surface.

The $\AdS_5$ calculation must be treated independently.  It uses $S^4$ harmonics and the radial equation Eq.\eqref{eq:adsFriedmann}, includes the active scalars of a regular Type-IIB truncation rather than only the universal axion--dilaton pair, and implements holographic boundary conditions that distinguish sources from normalizable data.  The quadratic variations of $S_{\rm GH}$, $S_{\rm ct}$, and the charge-fixing term must be included.  For the scalar-singular single-field solution class, a conventional Sturm--Liouville problem is not even well-defined until the singular endpoint and its admissible completion are specified.

For reproducibility, the two numerical problems should use different dimensionless variables:
\begin{align}
 \text{flat }4d:&\quad x=a/a_0,
 \qquad \dot a^2=1-x^{-4};\\
 \AdS_5:&\quad x=a/a_c,\quad \lambda=a_c/L,
 \qquad (a')^2=1+\lambda^2x^2-(1+\lambda^2)x^{-6}.
\end{align}
The minimum regression standard is therefore the following.  The flat code must reproduce the derived BPS endpoint factorization Eq.\eqref{eq:thresholdsquare20}, the radial identity Eq.\eqref{eq:radialidentity4}, the coefficient in Eq.\eqref{eq:gravcost4}, and the finite-box $R_{\max}^{-2}$ approach to the zero continuum threshold.  The AdS code must recover the $\lambda\to0$ local-flat limit of the geometry and verify the scalar-range obstruction Eq.\eqref{eq:Hrange} for the single Type-IIB coupling.  Both computations must demonstrate convergence under radial cutoff, grid refinement, and variation of the neck regulator, and they must report spectra separately for each ensemble and spacetime.

\section{Unitary completion, $\alpha$-sectors, and long-distance throat terms}
\label{sec:unitaritysewing}

The preceding sections concern classical charge-sector saddles and their endpoint Hessian. I now discuss a related but separate question.  If a small throat wormhole is replaced at long distances by a two-end insertion, what constrains the relative weights of the different parent-universe placements?  This section does not derive the Type-IIB coefficient matrix \(C^{ij}\).  That derivation requires the action, determinant, charge projection, zero-mode insertions, collective-coordinate measure, and contour. My claim is conditional: once a controlled calculation gives \(C^{ij}\), the different placements of the two end insertions are parts of a single expression of displaced quantum oscillator. 

Rubakov's lower-dimensional models fix the scope of the displaced oscillator \cite{NirovRubakov,RubakovBaby,RubakovLongRange}.  In those models, baby-universe splitting in a two-dimensional parent universe can be represented in an enlarged string theory description (not to be confused with Type IIB string theory I presently discuss) as emission of a light string state by a heavy macroscopic string.  The emitted object is not an ordinary particle moving inside the parent two-dimensional spacetime. It represents a separate component in a larger Hilbert space. Thus Rubakov models motivate the question of unitary completion, though he does not compute the Type-IIB axion--dilaton coefficient matrix. 

The general trace-preserving quantum channel is described b y the Feynman--Vernon influence functional \cite{FeynmanVernon1963}.  For a source \(J\) coupled to an eliminated sector \(E\), the influence functional is 
\begin{equation}
 {\cal F}[J_+,J_-]
 =\mathrm{Tr}_{E}\!\left(
 U_E[J_+] \, \rho_E \, U_E[J_-]^\dagger
 \right) .
 \label{eq:FVfunctional30}
\end{equation}
Unitarity gives
\begin{equation}
 {\cal F}[J,J]=1 .
 \label{eq:FVunitarity30}
\end{equation}
This is the trace preservation. In a linear Gaussian sector, Eq.~\eqref{eq:FVunitarity30} ties the same-history and cross-history terms into one package. This also holds in QED for soft photons forming a quantum environment to hard particle system~\cite{ReySoftOQS, ReySoftZX}. 

Bloch--Nordsieck and Pauli--Fierz systems are standard examples of this linear-source structure \cite{BlochNordsieck1937,PauliFierz1938}.  With the convention
\begin{equation}
 H[J]=\sum_\lambda \omega_\lambda b_\lambda^\dagger b_\lambda
 +\sum_\lambda\left(J_\lambda b_\lambda^\dagger+J_\lambda^*b_\lambda\right)+E_0[J],
\end{equation}
completion of the square gives
\begin{equation}
 c_\lambda=b_\lambda+{J_\lambda\over\omega_\lambda},
 \qquad
 H[J]=\sum_\lambda\omega_\lambda c_\lambda^\dagger c_\lambda
 +E_0[J]-\sum_\lambda{|J_\lambda|^2\over\omega_\lambda} .
\end{equation}
This oscillator algebra supplies a useful insight to the wormhole calculus. Generic soft photons are radiative modes with momentum, angle, polarization, gauge constraints, and memory data. In the strict zero momentum limit, these oscillators become a frozen, commuting source labels, as in Coleman's calculus. 

Coleman's calculus supplies the topology-specific version \cite{Coleman1988,Coleman1988b,HebeckerReview}.  The baby-universe Fock operators obey
\begin{equation}
 [a_i,a_j]=[a_i^\dagger,a_j^\dagger]=0,
 \qquad
 [a_i,a_j^\dagger]=\delta_{ij} .
\end{equation}
The operators that multiply local terms in the parent Hamiltonian are commuting Hermitian combinations, written schematically as
\begin{equation}
 A_i=a_i+a_i^\dagger,
 \qquad
 [A_i,A_j]=0 .
\end{equation}
They can be diagonalized,
\begin{equation}
 A_i|\alpha\rangle=\alpha_i|\alpha\rangle .
\end{equation}
At fixed \(\alpha\), the parent Hamiltonian is ordinary:
\begin{equation}
 {\cal H}_\alpha={\cal H}_0+\sum_i\alpha_i{\cal H}_i .
\end{equation}
Tracing over a diagonal distribution of \(\alpha\)'s gives a fixed-Hamiltonian average,
\begin{equation}
 \Phi(\rho)=\int d\alpha\,P(\alpha)\,U_\alpha\rho U_\alpha^\dagger .
\end{equation}
This is not a Haar-random circuit, a pseudorandom circuit, or a Lindblad bath.  It is an average over rigid coupling constants.

The Euclidean source algebra is direct.  Let \(I_i^{(A)}\) be the \(i\)-th end insertion on the component \(M_A\).  The total source coupled to the eliminated \(\alpha\)-sector is
\begin{equation}
 I_i=\sum_A I_i^{(A)} .
\end{equation}
For a Gaussian distribution of \(\alpha\)'s,
\begin{equation}
 \int d\alpha\,P(\alpha)\exp(\alpha_i I_i)
 =\exp\!\left({1\over2}I_iC^{ij}I_j\right) .
\end{equation}
Expanding the total source gives
\begin{equation}
 {1\over2}I_iC^{ij}I_j
 =\sum_{A<B}I_i^{(A)}C^{ij}I_j^{(B)}
 +{1\over2}\sum_A I_i^{(A)}C^{ij}I_j^{(A)} .
 \label{eq:diagonalplacement30}
\end{equation}
The first term has the two end insertions on different components.  The second has both end insertions on the same component.  They are not independent assumptions.  They are the mixed and equal-component terms produced by one coefficient matrix.  Removing the equal-component term while retaining the mixed term requires a charge projection, zero-mode saturation rule, boundary condition, supersymmetry constraint, or contour cancellation that acts on the same coefficient matrix.

The role of unitarity is evident. Feynman--Vernon, Bloch--Nordsieck, and Pauli--Fierz systems show how a linearly sourced eliminated sector must be treated as a complete trace-preserving quantum channel.  Coleman gives the topology-specific frozen-sector realization. The Type IIB string theory is still required for the calculation of \(C^{ij}\), including the charge lattice, fermion zero modes, collective coordinates, determinant, and contour.

I stress that the same-universe term in Eq.~\eqref{eq:diagonalplacement30} is not a local counterterm.  It is a bilocal term on one asymptotic parent universe.  At short separation, it can mix with local counterterms, but its long-distance origin is different: both end insertions come from the same small throat.  A Wilsonian subtraction may renormalize local operators; it does not remove the finite separated insertion unless the same microscopic calculation supplies a cancellation.

The discussion above does not claim that {\sl every} same-universe term survives in {\sl every} Type-IIB compactification.  Duality projections, charge conservation, supersymmetric zero modes, or boundary conditions can remove a particular operator.  The claim is structural.  Once a given coefficient matrix is derived, its mixed and equal-component placements are tied by the same source algebra.

\section{Relation to earlier and recent works}
\label{sec:relation}

The flat-space part of the paper isolates the charge-sector content of the older axionic-string and wormhole analysis \cite{Rey1991,Rey1989a,Rey1989b}. The main point is the ordering of the problem: the \(E=0\) BPS instanton comes first, and the \(E>0\) wormhole is a non-BPS deformation of it.  The physical Hessian at \(E=0\) inherits the Witten--Nester/Bogomolny adjoint-square form familiar from supergravity stability arguments \cite{Witten1981,Nester1981,Rey1992-1993}.

Recent discussions of axion wormholes emphasize boundary conditions, negative modes, and duality projections \cite{HertogTrigianteVanRiet,HertogTruijenVanRiet,HertogMaenaut,MarolfStability,AOP}.  The present paper separates two issues that are often conflated.  The adjoint-square Hessian is a statement about the BPS endpoint in the chosen charge sector after constraint reduction.  The complete \(E>0\) non-BPS Euclidean Hessian is a different operator.  It depends on the ensemble, the boundary form, the contour, and any additional fields in the compactification.

Coleman's baby-universe calculus is used here only for the topology-specific source algebra \cite{Coleman1988,Coleman1988b,HebeckerReview}.  Rubakov's lower-dimensional models~\cite{NirovRubakov,RubakovBaby,RubakovLongRange} and the Feynman--Vernon viewpoint explain why an eliminated sector should be treated as a complete trace-preserving object \cite{FeynmanVernon1963}.  Bloch--Nordsieck and Pauli--Fierz systems \cite{BlochNordsieck1937,PauliFierz1938} provide exactly solvable linear-source quantum oscillator examples, whose strict zero momentum limit is Coleman-like. 

\section{Conclusions}

The results of the paper have the following hierarchy.  The \(E=0\) solution is the BPS instanton.  For \(E>0\), the same charge-sector equations give non-BPS wormholes.  At \(E=0\), after the Hamiltonian constraint, gauge fixing, charge-sector boundary conditions, and zero-mode quotienting are imposed, the physical Hessian factorizes:
\begin{equation}
 \cH_\nu=\mathcal Q_\nu^\dagger\mathcal Q_\nu .
\end{equation}
This is the main endpoint theorem of the paper.

The theorem is not a claim of full stability for the \(E>0\) Euclidean wormhole.  The non-BPS Hessian has to be analyzed with its own ensemble, boundary form, contour, and matter content.  In \(\AdS_5\), the same endpoint logic is expressed through radial Sasaki--Mukhanov variables and a renormalized self-adjoint boundary problem.  The single axion--dilaton AdS solution has a smooth Einstein-frame throat but a scalar singularity; regular AdS wormholes require the actual scalar manifold and boundary data of the compactification.

The long-distance throat term is a separate issue.  If a small throat is represented by two end insertions with coefficient matrix \(C^{ij}\), then mixed-component and same-component terms come from the same quadratic expression.  Keeping one and deleting the other requires a real mechanism, such as a charge projection, zero-mode saturation rule, boundary condition, supersymmetry constraint, or contour cancellation.  Coleman supplies the topology-specific language for this statement; Feynman--Vernon and soft-radiation examples supply only the general lesson that an eliminated linear sector must be treated as one trace-preserving package.

\section*{Acknowledgments}
I acknowledge stimulating discussions with participants at 'Quantum PCP, Area Laws and Quantum Gravity' workshops held at the Institute for Pure \& Applied Mathematics (IPAM, USA) and the Simons Institute for the Theory of Computing (SIfTC, USA). This work was supported in part by the U.S. National Science Foundation and the Simons Foundation, and by the National Research Foundation of Korea (NRF) (RS-2021-NR060112) and by funds from Kwangwoon University. 

\appendix
\section{Constraint-reduced Sasaki--Mukhanov operator and the BPS Hessian factorization}
\label{app:ms-square}

\subsection{Four-dimensional charge-sector endpoint}
The factorization used in the main text concerns the reduced physical Hessian.  The reduction is stated explicitly to locate the conformal-factor objection in the correct space.  In the $O(4)$-invariant radial problem let
\begin{equation}
 a=a_0+f,
 \qquad
 D=D_0+g,
\end{equation}
where $D$ denotes the scalar variable dual to the fixed form charge, or the corresponding dilaton--axion coordinate in the form-field Routhian description.  Before imposing the constraint, the quadratic form has the block structure
\begin{equation}
 \delta^2S={1\over2}
 \begin{pmatrix} f & g \end{pmatrix}
 \begin{pmatrix}
  \cH_{aa} & \cH_{aD}\\
  \cH_{Da} & \cH_{DD}
 \end{pmatrix}
 \begin{pmatrix} f\\ g \end{pmatrix},
 \label{eq:appunreducedblock24}
\end{equation}
with the inner product including the radial measure and the angular degeneracy label.  Equivalently, in local notation it contains terms of the form
\begin{equation}
 \int\dd\tau\,
 \bigl(K_{aa}\dot f^2+K_{DD}\dot g^2+2K_{a\dot D}f\dot g
       +M_{aa}f^2+M_{DD}g^2+2M_{aD}fg\bigr).
\end{equation}
The lapse variation gives the linearized Hamiltonian constraint
\begin{equation}
 C_{\dot a}\dot f+C_a f+C_Dg+C_{\dot D}\dot g=0.
 \label{eq:appconstraint24}
\end{equation}
This equation is the Euclidean radial counterpart of the scalar constraint in the Lorentzian Sasaki--Mukhanov construction.  The indefinite block \eqref{eq:appunreducedblock24} is not the operator whose sign is being tested.  The physical scalar/form channel is obtained by solving \eqref{eq:appconstraint24}, together with the chosen gauge and charge-sector boundary data.  We write this solution abstractly as
\begin{equation}
 f=Tg,
 \label{eq:apptransfer24}
\end{equation}
where $T$ is the constraint transfer operator.  It may be local in a convenient gauge or nonlocal as a Green operator; its detailed form is not part of the positivity assumption.  Substitution gives the Schur-complement, or Euclidean Sasaki--Mukhanov, Hessian
\begin{equation}
 \cH_{\rm MS}
 =T^T\cH_{aa}T+T^T\cH_{aD}+\cH_{Da}T+\cH_{DD}.
 \label{eq:appMSschur24}
\end{equation}
This is the operator whose sign controls the BPS endpoint scalar/form perturbation.  A pure conformal fluctuation is a variation in the independent scale-factor direction of the unreduced action.  Except for gauge or collective-coordinate components that are quotiented separately, such a vector is not in the domain defined by \eqref{eq:appconstraint24}.  The reduced endpoint Hessian factorization therefore avoids the raw DeWitt conformal problem by construction: the conformal vector has been eliminated by the constraint and gauge projection before the operator whose sign is studied is formed.

At the BPS endpoint in the chosen charge sector the Routhian supergravity action is Bogomolny-completed,
\begin{equation}
 S-S_{\rm BPS}
 ={1\over2}\langle\mathcal B(\Phi),\mathcal B(\Phi)\rangle,
 \qquad
 \mathcal B(\Phi_0)=0,
 \label{eq:appBogo24}
\end{equation}
with $\Phi$ denoting the constrained Einstein--dilaton--form variables.  If $\mathcal D_{\nu,{\rm phys}}$ denotes the reduced charge-sector domain, the linearized operator is
\begin{equation}
 \cQ_\nu=\left.D\mathcal B\right|_{\Phi_0,\mathcal D_{\nu,{\rm phys}}},
 \qquad
 \mathcal B(\Phi_0+\delta\Phi_{\rm phys})
 =\cQ_\nu\delta\Phi_{\rm phys}+O(\delta\Phi_{\rm phys}^2).
\end{equation}
Therefore the reduced BPS endpoint operator is
\begin{equation}
 \cH_\nu=\cQ_\nu^\dagger\cQ_\nu,
 \label{eq:appQsquare24}
\end{equation}
where the adjoint is taken with the charge-sector radial inner product and the self-adjoint domain inherited from the constraint and boundary form.  The finite-box eigenvalue problem used in numerical checks is a discretization of \eqref{eq:appMSschur24} at the endpoint; it tests the discretization of the factorized operator and does not supply the derivation.

\subsection{AdS$_5$ radial version}
The same mechanism is present in the five-dimensional gauged-supergravity BPS endpoint problem.  For a radial metric
\begin{equation}
 \dd s_5^2=\dd\rho^2+a(\rho)^2\dd s^2_{\Sigma_4},
 \qquad
 \mathcal K={a'\over a},
\end{equation}
with scalar fields $\phi^a$ and target metric $G_{ab}$, the scalar perturbations $\varphi^a$ mix with the scalar trace perturbation $\psi$ of the induced metric.  The radial Hamiltonian and momentum constraints remove the lapse, shift, and pure trace direction.  The physical variables may be represented locally by the radial Sasaki--Mukhanov combinations
\begin{equation}
 \mathfrak a^a
 =\varphi^a-{\phi^{a\prime}\over\mathcal K}\psi,
 \label{eq:appadsMSvar24}
\end{equation}
in patches with $\mathcal K\neq0$, up to convention-dependent normalization of $\psi$.  At zeros of $\mathcal K$ one changes to an equivalent constraint-reduced coordinate.  The reduced quadratic action has the covariant form
\begin{equation}
 \delta^2S_{5,{\rm phys}}
 ={1\over2}\int\dd\rho\,a^4
 \left[
  G_{ab}D_\rho\mathfrak a^aD_\rho\mathfrak a^b
  +\mathfrak a^a\mathcal M_{ab}\mathfrak a^b
  +a^{-2}\lambda_{\rm slice}G_{ab}\mathfrak a^a\mathfrak a^b
 \right]
 +\delta^2S_{{\rm bdy},\ren}.
 \label{eq:appadsMSaction24}
\end{equation}
Here $D_\rho$ is the sigma-model covariant derivative and $\lambda_{\rm slice}$ is the eigenvalue of the transverse harmonic or boundary momentum.  For a BPS or fake-BPS radial flow the first-order equations define a linearized map
\begin{equation}
 (\cQ_5\mathfrak a)^a
 =D_\rho\mathfrak a^a+\mathcal U^a{}_b(\rho)\mathfrak a^b,
 \label{eq:appadsQ24}
\end{equation}
where $\mathcal U^a{}_b$ is the covariant Hessian of the first-order flow, including the warp-factor contribution and target-space connection.  The renormalized physical BPS endpoint operator is then
\begin{equation}
 \cH_{5,\nu}
 =\cQ_{5,\nu}^\dagger\cQ_{5,\nu}
 +a^{-2}\lambda_{\rm slice}
 +\cH_{{\rm bdy},\ren}.
 \label{eq:appadsQsquare24}
\end{equation}
The last term is part of the self-adjoint AdS problem: it encodes the finite counterterm variation and the allowed boundary condition.  Positivity is therefore a statement about the BPS-compatible renormalized domain, not about arbitrary mixed boundary conditions.  It is not an additional bulk channel competing with the factorized operator.  As in four dimensions, the theorem is a BPS endpoint statement.  It says that the BPS instanton/domain-wall endpoint has a factorized physical normal operator.  The \(E>0\) wormhole, if globally regular, has to be analyzed as a non-BPS deformation with its own boundary and contour problem.

\subsection{Boundary conditions and operator domain}
The ensemble is part of the operator.  In the scalar description one may impose boundary data on the dual axion variable; in the form/Routhian description one fixes the conserved charge.  These are related by a transform, but they do not define the same Hessian before the transform and boundary form are specified.  The factorization \eqref{eq:appQsquare24} is the BPS endpoint in the chosen charge sector statement.  Moving to a charge-varying scalar ensemble or to a neck-cut half-geometry changes the domain and may introduce modes that are not elements of the BPS endpoint Sasaki--Mukhanov sector.  This is why the theorem supplies the BPS endpoint operator without claiming a universal negative-mode count for every Euclidean saddle built from the same local fields.

\section{General dimensional scaling}

For a $D$-dimensional $O(D)$-invariant throat with
\begin{equation}
 (a')^2=1-\frac{a_0^{2D-4}}{a^{2D-4}}
\end{equation}
in the locally flat regime, the minimal-sphere area scales as $a_0^{D-2}$. If the first integral is normalized so that $E\propto a_0^{2D-4}$, then
\begin{equation}
 \Delta S_D\sim\frac{a_0^{D-2}}{\kappa_D^2}
 \propto E^{1/2}.
 \label{eq:genscale}
\end{equation}
If a different literature convention calls the squared charge parameter $E$, the same statement can appear as $E^{(D-2)/(2D-4)}=E^{1/2}$. Therefore small-neck exponents must be compared only after identifying the precise definition and mass dimension of the non-extremality parameter. The earlier shorthand $E^{1/(D-2)}$ is not convention-independent and should not be quoted without this check.

The scaling law is kinematic; the fluctuation statement is dynamical.  At the BPS endpoint the relevant charge-sector operator is not inferred from the power of $E$ but from the factorization of the constraint-reduced quadratic form.  In flat space this gives \eqref{eq:thresholdsquare20}; in AdS$_5$ the same statement is expressed by the radial Sasaki--Mukhanov operator \eqref{eq:adsQsquare23}, with curvature and boundary-domain effects included in the renormalized radial problem.  Thus the bottom of the spectrum is controlled by the $E=0$ BPS instanton endpoint, not by a small-neck scaling artifact.

\section{Boundary orientation at the neck}

Let $\cM=\cM_-\cup\cM_+$ be cut at $\Sigma$. The scalar boundary terms are
\begin{equation}
 \delta S_-\big|_\Sigma=\int_\Sigma\sqrt h\,n_-^\mu\partial_\mu D\,\delta D,
 \qquad
 \delta S_+\big|_\Sigma=\int_\Sigma\sqrt h\,n_+^\mu\partial_\mu D\,\delta D.
\end{equation}
Smoothness gives $n_+^\mu=-n_-^\mu$ as oriented normals to the two pieces and a single-valued normal derivative in the uncut manifold. Hence the sum vanishes. This cancellation is not available when one keeps only one half without fixing or Legendre-transforming the neck scalar datum.

\end{document}